\documentclass[11pt]{article}
\usepackage{times}
\usepackage[ dvips]{color}
\usepackage{amsmath}
\usepackage{amssymb}
\usepackage{amsthm}
\usepackage{epsfig}
\usepackage{url}

\usepackage{algorithm}
\usepackage{algorithmic}
\usepackage{verbatim}

\numberwithin{equation}{section}

\newtheorem{corollary}{Corollary}
\newtheorem{lemma}{Lemma}
\newtheorem{theorem}{Theorem}

\newtheorem{definition}{Definition}

\topmargin by -\headsep \textheight 8.9in \oddsidemargin 0pt
\evensidemargin \oddsidemargin \textwidth 6.5in
\floatstyle{ruled}
\newfloat{algo}{tbp}{lop}
\floatname{algo}{Algorithm}

\begin{document}

\date{}
\title{Simple Parallel and Distributed Algorithms \\ for Spectral Graph Sparsification \footnote{ This work is supported by NSF CAREER award CCF-1149048.}}
\author{Ioannis Koutis \\ Computer Science Department\\  University of Puerto Rico-Rio Piedras \\ioannis.koutis@upr.edu \\ }
\maketitle

\begin{abstract}
 We describe a simple algorithm for spectral graph sparsification, based on iterative computations of weighted spanners and uniform sampling. Leveraging the algorithms of Baswana and Sen for computing spanners, we obtain the first distributed spectral sparsification algorithm. We also obtain a parallel algorithm with improved work and time guarantees. 
 Combining this algorithm with the parallel framework of Peng and Spielman for solving symmetric diagonally dominant linear systems, we get a parallel solver which is much closer to being practical and significantly more efficient in terms of the total work.
\end{abstract}

\section{Introduction}

The efficient transformation of dense instances of graph problems to nearly equivalent sparse instances is a
powerful tool in algorithm design. Spectral sparsifiers are sparse graphs that preserve within an $1+\epsilon$ factor the quadratic form $x^TL_G x$, where $L_G$ is the Laplacian of $G$
 and $\epsilon$ is a parameter of choice. They were introduced by Spielman and Teng
\cite{SpielmanTeng04} as a basic component of the first nearly-linear time solvers
for linear systems  on symmetric diagonally dominant (SDD)
matrices  \footnote{A symmetric matrix $A$ is SDD if for all $i$, $A_{ii} \geq \sum_{j\neq i} |A_{ij}|$.}. Such linear system solvers are a key algorithmic primitive with numerous applications \cite{KMP2012,Teng:2010}.

The Spielman and Teng sparsification algorithm produces sparsifiers with
$O(n\log^c n/\epsilon^2)$ edges for some fairly large constant $c$, where $n$
is the number of vertices in the graph. At a high level
their algorithm is based on graph decompositions into edge-disjoint sets
that get sparsified independently via uniform sampling.
As noted in \cite{PengSpielman13} the algorithm can be parallelized
if the original partitioning subroutine is substituted by a more
recent one due to Orecchia and Vishnoi \cite{DBLP:conf/soda/OrecchiaV11}.

Peng and Spielman \cite{PengSpielman13} recently presented a novel algebraic framework
for solving SDD systems. It enables the use of parallel sparsification algorithms
for constructing parallel solvers. Combined with the parallelized Spielman and Teng sparsification algorithm,
or a more recent approach due to Peng (Section 3.4, \cite{Peng.thesis}), this algebraic
 framework yields the first `truly' parallel SDD solver that does near-linear work and runs in polylogarithmic time.

The new parallel solver leaves something to be desired:
its work is by several logarithmic factors larger than that of the fastest known sequential
algorithm that runs in $\tilde{O}(m\log n)$ time\footnote{We use $\tilde{O}()$
to hide a $poly(\log \log n)$ factor.}; here $m$ is
 the number of non-zero entries in the matrix \cite{KoutisMP_FOCS11}. 
This motivates our study on parallel and distributed sparsification
algorithms.

{\bf Background on spectral sparsification.} Besides yielding the SDD solver, the work of Spielman and Teng
spurred further research on spectral sparsification
as a stand-alone problem.  Spielman and Srivastava~\cite{SpielmanS08}
showed that it is possible to produce a sparsifier
with $O(n \log n/\epsilon^2)$ edges in near-linear time.  
Their approach is based on viewing the graph as an electrical
resistive network, where one can define the effective resistance
of an edge as the potential difference that must be applied
between its two endpoints in order to send one unit of electrical
flow from the one vertex to the other. The sparsifier is computed by sampling edges
with probabilities proportional to the their effective resistances.  
Spielman and Srivastava also showed that $O(\log n)$ calls
to a solver for SDD linear systems can produce sufficiently good
approximations to all effective resistances, allowing
for a near-linear time implementation of their sampling scheme.
This development was followed by works
on slower but more sparsity-efficient spectral sparsification
algorithms \cite{BatsonSS09,Kolla10} and on sparsification
in the semi-streaming model \cite{KelnerL11}.

The work of Spielman and Srivastava opened the way to 
the near-$m\log n$ time solver in \cite{KoutisMP_FOCS10, KoutisMP_FOCS11}.
This fast solver utilizes an `incremental sparsification'
algorithm that produces a very mildly sparser spectral approximation
to the input graph. A direct by-product of this fast solver was the acceleration
of the Spielman-Srivastava sparsification scheme. Their
scheme was further improved in \cite{KoutisLP12,KoutisLP12a},
yielding an $\tilde{O}(m)$ solver for slightly non-sparse graphs; the
 solver combines in an intricate recursive way slower solvers
with spectral sparsifiers.

Recent  efforts aim to obtain simpler algorithms via alternative
approaches. In particular, there has been an interest in
combinatorial algorithms that rely less on the power
of algebra to achieve similar results \cite{KapralovP2012,Kelner2013}.
We do not insist that these simpler algorithms are asymptotically as efficient
as their algebraic counterparts. In practice there are
many phenomena, subtler than asymptotic behavior or even hidden
constants, that affect the performance of linear system solvers, and
different ideas may lead to better implementations. In particular,
there are implementations that exhibit great empirical
performance on sparse matrices  \cite{KoutisMiller09,LivneB12};
solve-free techniques for spectral sparsification
have the potential of extending the applicability of these
implementations to dense matrices. 

The first combinatorial alternative to the spectral sparsification algorithm
of Spielman and Teng was given by Kapralov and Panigrahi \cite{KapralovP2012}.
A novel feature of their work is the introduction of spanners
in the context of spectral graph sparsification. 
The algorithm is based on tightly approximating
effective resistances; more concretely, they define
`robust connectivities' of edges and show they are good
upper bounds to the effective resistances, on average. Approximate
robust connectivities are then used for sampling; the
result follows from an application of the `oversampling'
Lemma of~\cite{KoutisMP_FOCS10} which shows
that extra sampling can compensate for the the
lack of accuracy in the estimates for the effective
resistances; this extra sampling yields the slightly
more dense sparsifier. 
The algorithm generates a sparsifier with $O(n\log^4 n/\epsilon^4)$ edges in $O(m\log^4 n)$ time
but it doesn't parallelize mostly due to the use of distance
oracles by Thorup and Zwick~\cite{ThorupZ05}.

For a more thorough review of the sparsification literature, 
we refer the reader to the excellent article by Batson et al.~\cite{BatsonSST13}.


\textbf{In this work.} We describe a simple parallel and distributed algorithm
that exposes  a closer connection between spanners and sparsification. 
Using only iterated computations of weighted spanners and uniform sampling the
algorithm produces an $(1\pm \epsilon)$-approximation with
$O(n\log^3 n \log^3 \rho/\epsilon^2 + m/\rho)$ edges, where $\rho$ is \textbf{sparsification factor}
of choice.

The idea behind the algorithm is simple.  In order to reduce
the number of edges by a factor of~$\rho$, we compute $O(\log^2 n \log^2 \rho/\epsilon^2)$ edge-disjoint
spanners of the graph that allow us to certify upper bounds for
the effective resistances of the rest of the edges.
The upper bounds enable uniformly sampling-away
about half of the remaining edges while spectrally preserving
the graph within a $(1+\epsilon/(4 \log \rho))$ factor. The process
is applied iteratively, and after $O(\log \rho)$ rounds we get a graph
that $(1+\epsilon)$-approximates the input graph and has $O(n\log^3 n \log^3 \rho + m/\rho)$
edges. The total work is $O(m \log^2 n \log^3 \rho/\epsilon^2)$.

We use our parallel sparsification algorithm to obtain a solver for
SDD linear systems that works in polylogarithmic time and does
$\tilde{O}((m\log^2 n + n\log^{5} n \log^5 \kappa)(\log(1/\tau) )$ work,
where $\tau$ is a standard measure of tolerance in the error
of the approximate solution, and $\kappa$ is the condition number
of the input system.

\section{Background}

\textbf{Laplacians.} Given a weighted graph $G = (V,E,w>0)$  where $V=\{1,\ldots,n\}$, its Laplacian $L_G$ is the
matrix defined by: 
$$ \textnormal{(i) $L_G(i,j) = -w_{ij}$ for $i\neq j$ ~and~ (ii) $L_G(i,i)=\sum_{j\neq i} w_{ij}$}.$$

Throughout the paper we will $n,m$ to denote the number of vertices and edges of a graph respectively. We will apply algebraic operators on graphs in a standard way. Specifically, given two graphs $G_1= (V,E,w_1)$ and $G_1= (V,E,w_2)$ we denote by $G_1+G_2$ the graph $(V,E,w_1+w_2)$. Also given a scalar $a$ we let $aG_1 = (V,E,a w_1)$.

\textbf{Spectral approximation.} We say that a graph $H$, $(\beta/\alpha)$-approximates a graph $G$ if:
$$
     \alpha(x^T L_H x) \leq x^T L_G x \leq \beta (x^T L_H x).
$$

Finally, if for all vectors $x$ we have $x^T L_{G_2} x \leq x^T L_{G_1} x$ we will write $G_2\preceq G_1$.

\textbf{Stretch.}  Let $p$ be a path joining the two endpoints of an edge $e\in E$.  The stretch $st_p(e)$ of an edge $e$, is equal to
$$
    w_e\sum_{e'\in p} (1/w_{e'}).
$$
We also define the stretch of $e$ over a graph $H$ as
$$
    st_H(e) = \min_{p\in H} st_p(e).
$$

\textbf{Spanners.} A $\log n$-spanner of a graph $G$ is a subgraph $H$ of $G$ such that for all edges $e\in E$
$$
    st_H(e) \leq 2\log n.
$$

In the rest of the paper we will use the term spanner to mean a $\log n$-spanner.
Every graph contains a spanner with $O(n\log n)$ edges that can be computed efficiently in the CRCW PRAM model and the synchronous distributed model. Concretely, we adapt here Theorems 5.4 and 5.1 respectively, from Baswana and Sen \cite{BaswanaS07}.

\begin{theorem}  \label{th:parSpanner}
Given a graph $G$, a spanner for $G$ of expected size $O(n\log n)$ can be constructed with $O(m\log n)$ work in $\tilde{O}(\log n)$ time with high probability. The algorithm runs in the CRCW PRAM model.
\end{theorem}

\begin{theorem} \label{th:distrSpanner}
Given a graph $G$, a spanner for $G$ of expected size $O(n\log n)$ can be constructed in the synchronous distributed model in $O(\log^2 n)$ rounds and $O(m\log n)$ communication complexity. Moreover, the length of each message communicated is $O(\log n)$.
\end{theorem}

%

Here we define an object that plays a key role in our algorithm.

\begin{definition}
 Let $G$ be a graph and $H_1,\ldots,H_t$ be subgraphs of $G$ such that $H_i$ is a spanner for the graph
 $G-\sum_{j=1}^{i-1} {H_j}$. We call $H = \sum_{j=1}^t H_j$ a $t$-bundle spanner. We call the $H_i$'s the components of $H$.
\end{definition}

\textbf{Effective Resistance.} A graph can be viewed as an electrical resistive network, with each edge corresponding to a resistor having resistance $r_e = 1/w_e$. The effective resistance $R_{u,v}[G]$ between two  vertices $u$ and $v$ in $G$  is defined as the potential difference that has to be applied on $u$ and $v$ in order to drive one unit of current through the network. For instance, in the case of a path $p$ the effective resistance between the two endpoints of $p$
is equal to $R_e[p]= \sum_{e'\in p} (1/w_{e'})$; this is the well known formula for resistors connected in series.

Now let us recall a simple fact about paths connected `in parallel', i.e. paths that are vertex-disjoint with the exception of their shared endpoints $u$ and $v$. Let $p_1,\ldots,p_t$ be paths connected in parallel. Let $P = \sum_{i=1}^t p_i$.
For the effective resistance between $u$ and $v$, in the graph $P$ consisting of the union of the paths, we have
\begin{equation} \label{eq:parallel}
   R_{u,v}[P] = \left(\sum_{i=1}^t (R_{u,v}[p_i])^{-1}\right)^{-1}.
\end{equation}

The following Lemma has a key role in our sparsification algorithm.

\begin{lemma} \label{lem:bundle}
 Let $G$ be a graph and $H$ be a $t$-bundle spanner of $G$.  For every edge $e$ of
 $G$ which is not in $H$, we have
 $$
      w_e R_e[G] \leq  \log n /t.
 $$
\end{lemma}
\begin{proof}
Let $H_1,\ldots,H_t$ be the components of $H$.
If $H'$ is any subgraph of $G$ then by Rayleigh's monotonicity law \cite{doyle-2000} the effective resistance of $e$ is at most equal to the effective resistance between the two endpoints of $e$ in $H'$.
In particular, fix an arbitrary edge $e$ not in $H$. For each $i$  we know by definition that it contains a path $p_i$ such that
$$
    w_e\sum_{e'\in p_i} (1/w_{e'}) \leq 2\log n.
$$
As we discussed above $\sum_{e'\in p_u} (1/w_{e'})$ is equal to the resistance between the two endpoints of $e$ in $p$. This implies that the effective resistance of $e$ over $p_i$ satisfies
$$
    R_{e}[p_i] \leq 2\log n/w_e.
$$
Now we observe that by definition the paths $p_i$ connect in parallel the two endpoints of $e$. Let $P = \sum_{j=1}^t p_i$. By invoking equality \ref{eq:parallel} and combining with the last inequality we get that
\begin{eqnarray*}
     (R_e[P])^{-1} & = & \left(\sum_{i=1}^t (R_{e}[p_i])^{-1}\right) \\ & \geq &  t w_e/(2\log n). 
\end{eqnarray*}     
which implies
$$ R_e[P] \leq \log n /(t w_e).$$
Finally, we have $R_e[G] \leq R_e [P]$ by Rayleigh's monotonicity law, since $P$ is a subgraph of $G$.
\end{proof}

Let $B_e$ be the $n\times n$ Laplacian of the unweighted edge $e$ (which is zero everywhere except a 2x2 submatrix). Looking at the
effective resistance algebraically, it is well understood that:
$$
     B_e \preceq R_e[G] G.
$$

Then the above lemma implies the following.

\begin{corollary} \label{cor:bundle}
 Let $G$ be a graph and $H$ be a $t$-bundle spanner of $G$.  For every edge $e$ of
 $G$ which is not in $H$, we have
 $$
     w_e B_e  \preceq  \frac{\log n}{t} G.
 $$
 \end{corollary} 

\section{Parallel Sparsification}

\subsection{Parallel $t$-bundle Spanner Construction}

A $t$-bundle spanner can be computed iteratively in the obvious way: in the $i$th iteration we compute a spanner
$H_i$ for $G-\sum_{j=1}^{i-1} {H_j}$. Edges in $\sum_{j=1}^{i-1} H_j$ can declare themselves out of the $i$th iteration in the parallel or distributed model. Thus extending the algorithms of Baswana and Sen is easy, and we get the following corollaries.

\begin{corollary}  \label{cor:parSpanner}
On input of a graph $G$, a $t$-bundle spanner for $G$ of expected size $O(t n \log n)$ can be constructed with $O(t m\log n)$ work in $\tilde{O}(t \log n)$ time, with high probability. The algorithm runs in the CRCW PRAM model.
\end{corollary}

\begin{corollary} \label{cor:distrSpanner}
On input of a graph $G$, a $t$-bundle spanner for $G$ of expected size $O(t n\log n)$ can be constructed in the synchronous distributed model in $O(t \log^2 n)$ rounds and $O(t m\log n)$ communication complexity. Moreover, the length of each message communicated is $O(\log n)$.
\end{corollary}

\subsection{Sampling for Parallel Sparsification}

We will sparsify graphs using sampling. The Spielman-Srivastava scheme fixes the number of samples and for each sample one edge is selected according to a fixed probability distribution and gets added to the sparsifier~\cite{SpielmanS08}. In Algorithm\ref{alg:parallelsample} we use a slightly different sampling scheme, sampling each edge independently with a fixed probability.

\begin{algo} [h]
\underline{Input:} Graph $G$, parameter $\epsilon$

\underline{Output:} Graph $\tilde{G}$
\vspace{0.2cm}

\begin{algorithmic}[1]
\STATE{Compute a $(24 \log^2 n / \epsilon^2)$-bundle spanner $H$ for $G$}
\STATE{Let $\tilde{G}:=H$}
\STATE{For each edge $e \not \in H$ with probability $1/4$ add $e$ to $\tilde{G}$ with weight $4w_e$}
\STATE{Return $\tilde{G}$}
\end{algorithmic}
\caption{\textsc{ParallelSample}}\label{alg:parallelsample}
\end{algo}

We will need  a Theorem due to Tropp \cite{tropp2012user}, and more specifically its
following variant \cite{Harvey12}.

\begin{theorem} \label{th:tropp}
Let $Y_1,\ldots,Y_k$ be independent positive semi-definite matrices of size $n\times n$.
Let $Y = \sum_i Y_i$. Let $Z=E[Y]$. Suppose $Y_i\preceq R Z$. Then for all $\epsilon \in [0,1]$
\begin{eqnarray*}
Pr\left[\sum_i Y_i \preceq (1-\epsilon) Z\right] & \leq  & n \cdot exp(-\epsilon^2 /2 R)  \\
Pr\left[\sum_i Y_i \succeq (1+\epsilon) Z\right] & \leq & n \cdot exp(-\epsilon^2 /3 R).
\end{eqnarray*}
\end{theorem}

We have the following Theorem.

\begin{theorem} \label{th:parallelsample}
The output $\tilde{G}$ of algorithm \textsc{ParallelSample} on input $G$ and $\epsilon$ satisfies
with probability $1-1/n^2$ the following:
\begin{itemize}
\item [(a)] $ (1-\epsilon) G \preceq \tilde{G} \preceq (1+\epsilon )G.$
\item [(b)] The expected number of edges in $\tilde{G}$ is at most $$O(n\log^3 n/\epsilon^2+m/2).$$
\end{itemize}
\textsc{ParallelSample} can be implemented in the CWCR PRAM model to use $O(m\log^3 n /\epsilon^2)$ work in $\tilde{O}(\log^3 n/\epsilon^2)$ time. In the synchronous distributed model, \textsc{ParallelSample} can be implemented to run in $O(\log^4 n/\epsilon^2)$ rounds, with $O(m\log^3 n/\epsilon^2)$ communication complexity, using messages of size $O(\log n)$.
\end{theorem}

\begin{proof}
The work, parallel time, and communication complexity guarantees for \textsc{ParallelSample} follow directly from the Corollaries~\ref{cor:parSpanner} and \ref{cor:distrSpanner}, by letting $t=O(\log^2 n/\epsilon^2)$.

Now let $B_e$ be the $n\times n$ Laplacian of the unweighted edge $e$. For each edge $e\not \in H$ we let $Y_e$ be the random variable defined as follows:
\begin{eqnarray*}
Y_e  & =  & 0,  \qquad \qquad \textnormal{with probability $3/4$}, \\
     &  = &  4 w_e B_e  \qquad ~\textnormal{with probability $1/4.$}
\end{eqnarray*}
Also we let $$H_i = \lfloor \epsilon^2 /(6\log n) \rfloor H,$$ for $i=1,\ldots, (\lfloor \epsilon^2 /(6\log n) \rfloor)^{-1}$.
We apply Theorem \ref{th:tropp} to the random matrix that is formed by summing the $H_i$'s and the $Y_e$'s. For the output of the algorithm, we clearly have
$$
    \tilde{G} = \sum_{e\not \in H} Y_e + \sum_{i} H_i = \sum_{e\not \in H} Y_e+ H.
$$
We also have that $E[\tilde{G}] = G$. Using $H \preceq G$, for each $i$ we have
$$H_i = \lfloor \epsilon^2 /(6\log n) \rfloor H \preceq  \epsilon^2 /(6\log n)  G.$$
In addition for each $e\not \in H$, we have
$$
    Y_i \preceq 4w_e B_e \preceq \epsilon^2 /(6\log n) G.
$$
The last inequality follows by setting $t = 24\log^2 n/\epsilon^2$ in Corollary \ref{cor:bundle}.
Thus the condition of Theorem \ref{th:tropp} is satisfied for $R = \epsilon^2 /(6\log n)$, which substituted in the bounds of the Theorem proves that (a) holds with probability at least $1-1/2n^2$.
For (b), observe that the expected number of edges in $H$ is $O(n\log^3 n/\epsilon^2)$ as stated in Corollaries~\ref{cor:parSpanner} and \ref{cor:distrSpanner}. The expected numbers of edges outside $H$ is $m/4$ and a simple application of Chernoff's inequality implies that the number is at most $m/2$ with probability at least $1-1/2n^2$. Hence a union bound gives that both (a) and (b) hold with probability at least $1-1/n^2$.
\end{proof}

\subsection{The Algorithm}

The main sparsification routine is presented in Algorithm~\ref{alg:parallelsparsify}.

\begin{algo} [h]

\underline{Input:} Graph $G$, parameters $\epsilon,\rho$ \\
\underline{Output:} Graph $\tilde{G}$
\vspace{0.01cm}

\begin{algorithmic}[1]
\STATE{Set $G_0: = G$}
\STATE{For $i=1:\lceil \log \rho \rceil$}
\qquad \STATE{Set $G_i: = \textsc{ParallelSparsify}(G_{i-1},\epsilon/\lceil \log \rho \rceil)$}
\STATE{Return $G_{\lceil \log \rho \rceil }$}
\end{algorithmic}
\caption{\textsc{ParallelSparsify}}\label{alg:parallelsparsify}
\end{algo}

We prove the following Theorem.
\begin{theorem} \label{th:parallelsparse}
The output $\tilde{G}$ of algorithm \textsc{ParallelSparsify} on input $G$ and $\epsilon,\rho$ satisfies
$$
    (1-\epsilon) G \preceq \tilde{G} \preceq (1+\epsilon) G
$$
with high probability. The expected number edges in $G$ is at most
$$
   O(n \log^3 n \log^3 \rho/\epsilon^2 + m/\rho).
$$
The algorithm does $O(m\log^2 n \log^3 \rho /\epsilon^2)$ work and runs in $O(\log^3 n \log^3 \rho /\epsilon^2)$ time in the CRCW mode. In the synchronous distributed model, it can be implemented to run in $O(\log^4 n \log^3 \rho/\epsilon^2)$ rounds with $O(m \log^3 n \log^3 \rho/\epsilon^2)$ communication complexity, using messages of size $O(\log n)$.
\end{theorem}
\begin{proof}
We can show using induction and Theorem \ref{th:parallelsample}
that graph $G_t$ satisfies
$$
   (1-\epsilon/\log \rho)^t G \preceq  G_t \preceq (1+\epsilon/\log \rho)^t G.
$$
with probability $(1-1/n^2)^t$ and the expected number of edges in it
is at most $$O(n t\log^3 n \log^2 \rho /\epsilon^2 + m/2^t).$$
Since $t\leq \lceil \log \rho \rceil $, we get the desired spectral inequality.
The parallel and distributed implementations are straightforward. The
total work (and communication complexity) is dominated by the work performed in the first iteration,
since the size of the graphs decrease geometrically. The claims
on the parallel and distributed implementations then follow from Theorem~\ref{th:parallelsample}.
\end{proof}

\section{Improved parallel SDD solver}

\textbf{The Peng-Spielman parallel framework.} Peng and Spielman \cite{PengSpielman13} gave the first solver for symmetric diagonally dominant (SDD) linear system that does near-linear work in polylogarithmic time. We shortly review the basic ideas behind their solver in order to highlight how our sparsification routine can be plugged into it, thus deriving  work and time guarantees for a more efficient solver.

Let $D$ be a diagonal matrix and $A$ be the adjacency matrix of a graph with positive weights. The main idea in~\cite{PengSpielman13} is a reduction of the input SDD linear system with matrix $M_1 = D-A$, to a linear system with matrix $\tilde{M_1} = D-AD^{-1}A$ which is also shown to be SDD. Matrix $\tilde{M}_1$ is actually never formed explicitly because it can be too dense, as all vertices that are within a distance of 2 in graph $A$ form now a clique in graph $AD^{-1}A$. The \textbf{first} step to remedying this problem is replacing $\tilde{M_1}$ with a $(1+\epsilon/2)$-approximation ${\hat{M_1}}$ that has $O(n+ m\log n/\epsilon^2)$ edges and doesn't contain these cliques, but replaces them with sparse graphs. As shown in Corollary 6.4 of \cite{PengSpielman13} this can be done in in $O(\log n)$ time and $O(n+m\log^2 n/\epsilon^2)$ work. The \textbf{second} step is further sparsifying $\tilde{M_1}$ down to $O(n\log^c n/\epsilon^2)$ non-zeros (for some fairly large constant $c$), using the parallelized Spielman-Teng sparsification algorithm. This step forms a matrix $M_2$ which is a $(1+\epsilon)$-approximation of $\tilde{M_1}$, and also an SDD matrix which is of the form $D'-A'$.

This construction is repeated recursively, producing an `approximate inverse chain' for $M_1$: $$\{M_1,M_2,\ldots,M_d\}.$$ The depth $d$ of the chain needs to be $O(\log \kappa)$ where $\kappa$ is the condition number of $M_1$, i.e. the ratio of its largest to its smallest non-zero eigenvalue. This is because for $d=O(\kappa)$ the condition number of $M_d$ is very close to $1$, i.e. $M_d$ is essentially the identity matrix, and no further reductions are required.  The $(1+\epsilon)$ approximations incurred by the construction of $M_{i+1}$ from $M_i$ compound in a multiplicative fashion. So, in order to keep the total approximation bounded we need to pick, $\epsilon = 1/ O(\log \kappa)$.

As shown in Theorem 4.5 of \cite{PengSpielman13} an approximate inverse chain can be used to produce an approximate solution for the system in $O(d\log n)$ depth and total work proportional to the total number of non-zero entries in the matrices that constitute the chain.

\textbf{The solver.} We now outline the construction of a parallel SDD solver that uses our improved parallel sparsification algorithm. We can think of all matrices in the approximate inverse chain as Laplacians, and we will refer to them as graphs. For simplicity, we will use $\tilde{O}$ to suppress polylogarithmic factors in $n$ and~$\kappa$. Also, we note that the spectral approximation bounds hold with high probability, and the claims on the number of edges of the sparsifiers hold in expectation; we won't further discuss randomization for the sake of brevity.

Recall that in the construction of the approximate inverse chain, one has to set $\epsilon = 1/O(\log \kappa)$. Given that, observe also that the `threshold of applicability' of Theorem \ref{th:parallelsparse} is when the graph $M_i$ has more than $\tilde{O}(n\log^3 n \log^2 \kappa)$ edges, whenever the sparsification factor $j$ is of polylogarithmic size. Let us denote by $m'$ this threshold.  Whenever sparsification of $\tilde{M}_i$ is not possible, we simply let $M_{i+1}= \tilde{M}_i$, as implicitly done in \cite{PengSpielman13}.

When constructing $M_{i+1}$ from $M_i$, the number of edges goes up by a factor of $O(\log n \log^2 \kappa)$, in the first step that constructs $\tilde{M}_i$. In order to keep the total size of the inverse approximate chain and thus the work of the solver bounded, we only need to bring the graph back to its original size, if it exceeds $m'$. Besides its stronger guarantees, a relative advantage of our routine is that we can use it to sparsify the input graph by any factor~$\rho$, rather than aim for a very sparse graph as Peng and Spielman \cite{PengSpielman13} propose. So, using Theorem~\ref{th:parallelsparse} the graph can be sparsified down to ${O}(m' +m)$ edges, by setting $\rho=O(\log n \log^2 \kappa)$. The total work is $\tilde{O}((m'+m) \log^2 n \log^2 \kappa)$. Hence the total size of the approximate inverse chain is $\tilde{O}((m'+m)  \log \kappa )$, and the total work required for its construction is $\tilde{O}((m'+m) \log^2 n \log^3 \kappa)$.

We can improve the dependence on $m$ by constructing the chain not for the input matrix $M$, but for a $2$-approximation $M'$ of it, which has $\tilde{O}(n\log^3 n + m/\log^2 n \log^3 \kappa)$ edges. This can be constructed by invoking Theorem \ref{th:parallelsparse}, with $\epsilon = 1/2$ and $\rho = O(\log^2 n \log^3 \kappa)$. The total work for this step is $\tilde{O}(m \log^2 n)$. It is well understood that this approximate chain for $M'$ can be used as a preconditioner for $M$ (in the same way its own chain would be used) incurring only a constant factor in the work and time guarantees.

Combining the above with Theorem 4.5 of  \cite{PengSpielman13}, we get the following Theorem.

\begin{theorem}
 On input of a linear system $Mx = b$, where $M$ is an SDD matrix of dimension $n$ with $m$ non-zeros, a vector $x'$ that satisfies $||b - M^+ x||_M < \epsilon$ can be constructed  with probability at least $1/2$ in polylogarithmic time and $\tilde{O}(m\log^2 n + m' \log^5n \log^5 \kappa)$ work.
\end{theorem}

\section{Concluding Remarks}

{\bf Remark 1.} Multigrid algorithms
provably do linear work in
logarithmic time, for certain very special classes of SDD
systems that arise from the discretization of partial
differential equations \cite{Bramble93}.
The algebra underlying multigrid
is quite different than that used by Peng and Spielman;
in contrast with their algorithm, the spectral
approximation does not accumulate multiplicatively in the multigrid `chain'.
This imposes a much less demanding constraint for the
approximation quality between two subsequent levels,
which can be constant, rather than $O(1/\log \kappa)$. Much
of the efficiency of these specialized multigrid algorithms
stems from this fact. It  remain open whether something similar is possible for
general SDD matrices, In particular, it is still open
whether there is an $O(n)$-work $O(\log n)$ time
algorithm for regular weighted two-dimensional grids that
are `affinity' graphs of images. Experimental evidence \cite{KrishnanFS13}
seems to suggest that the possibility cannot be dismissed.

{\bf Remark 2.} It can be shown that low-stretch trees
can replace spanners in our construction, reducing the
size of the sparsifiers by an $O(\log n)$ factor. 
 The potential advantage of
such an algorithm would be that it provides a sparsifier
which is expressed naturally as a sum of trees.

{\bf Remark 3.} While a significant improvement over
the solver presented in \cite{PengSpielman13}, the total work of our parallel algorithm
remains high (in terms of the logarithmic factors) especially
for sparse graphs. We conjecture that more improvements
are possible, and will probably have to use a different
algebraic framework (see Remark~1). Within the Peng and
Spielman framework, it seems plausible that improvements
can come from replacing the $t$-bundle
 by a sparser object; this presents
us an interesting problem. The number of logarithmic
factors can be probably somewhat decreased by reducing
the dimension $n$, potentially by using a two-level
`Steiner preconditioning' scheme \cite{KoutisMiller08}.

{\bf Remark 4.} We wish emphasize the simplicity and
implementability of our algorithm as a stand-alone sparsification routine,
relative to the other two known solve-free algorithms by
Spielman and Teng \cite{SpielmanTeng04} and Kapralov
and Panigrahi \cite{KapralovP2012}. Comparing
to the latter, our algorithm has also the
`right' dependency on $\epsilon$ ($1/\epsilon^2$ vs $1/\epsilon^4$)
and is flexible with the sparsification factor $\rho$.


\bibliographystyle{plain}

\begin{thebibliography}{10}

\bibitem{BaswanaS07}
Surender Baswana and Sandeep Sen.
\newblock A simple and linear time randomized algorithm for computing sparse
  spanners in weighted graphs.
\newblock {\em Random Struct. Algorithms}, 30(4):532--563, 2007.

\bibitem{BatsonSS09}
Joshua~D. Batson, Daniel~A. Spielman, and Nikhil Srivastava.
\newblock {Twice-Ramanujan sparsifiers}.
\newblock In {\em Proceedings of the 41st Annual ACM Symposium on Theory of
  Computing}, pages 255--262, 2009.

\bibitem{BatsonSST13}
Joshua~D. Batson, Daniel~A. Spielman, Nikhil Srivastava, and Shang-Hua Teng.
\newblock Spectral sparsification of graphs: theory and algorithms.
\newblock {\em Commun. ACM}, 56(8):87--94, 2013.

\bibitem{Bramble93}
James~H. Bramble.
\newblock {\em Multigrid Methods}.
\newblock Chapman and Hall, 1993.

\bibitem{doyle-2000}
Peter~G. Doyle and J.~Laurie Snell.
\newblock Random walks and electric networks, 2000.

\bibitem{Harvey12}
N.~Harvey.
\newblock {Matrix Concentration}.
\newblock
  \url{http://www.cs.rpi.edu/~drinep/RandNLA/slides/Harvey_RandNLA@FOCS_2012.pdf},
  2012.

\bibitem{KapralovP2012}
Michael Kapralov and Rina Panigrahy.
\newblock Spectral sparsification via random spanners.
\newblock In {\em Proceedings of the 3rd Innovations in Theoretical Computer
  Science Conference}, ITCS '12, pages 393--398, New York, NY, USA, 2012. ACM.

\bibitem{KelnerL11}
Jonathan~A. Kelner and Alex Levin.
\newblock Spectral sparsification in the semi-streaming setting.
\newblock In {\em Proceeding of the 28th International Symposium on Theoretical
  Aspects of Computer Science, STACS}, pages 440--451, 2011.

\bibitem{Kelner2013}
Jonathan~A. Kelner, Lorenzo Orecchia, Aaron Sidford, and Zeyuan~Allen Zhu.
\newblock {A Simple, Combinatorial Algorithm for Solving SDD Systems in
  Nearly-Linear Time}.
\newblock {\em CoRR}, abs/1301.6628, 2013.

\bibitem{Kolla10}
Alexandra Kolla, Yury Makarychev, Amin Saberi, and Shang-Hua Teng.
\newblock Subgraph sparsification and nearly optimal ultrasparsifiers.
\newblock In {\em Proceedings of the 42nd ACM Symposium on Theory of Computing,
  (STOC)}, pages 57--66, 2010.

\bibitem{KoutisLP12a}
Ioannis Koutis, Alex Levin, and Richard Peng.
\newblock Faster spectral sparsification and numerical algorithms for sdd
  matrices.
\newblock {\em CoRR}, abs/1209.5821, 2012.

\bibitem{KoutisLP12}
Ioannis Koutis, Alex Levin, and Richard Peng.
\newblock Improved spectral sparsification and numerical algorithms for {SDD}
  matrices.
\newblock In {\em Proceedings of the 29th International Symposium on
  Theoretical Aspects of Computer Science, STACS}, pages 266--277, 2012.

\bibitem{KoutisMiller09}
Ioannis Koutis and Gary Miller.
\newblock The combinatorial multigrid solver.
\newblock Conference Talk, March 2009.

\bibitem{KoutisMiller08}
Ioannis Koutis and Gary~L. Miller.
\newblock {Graph partitioning into isolated, high conductance clusters: Theory,
  computation and applications to preconditioning}.
\newblock In {\em Symposiun on Parallel Algorithms and Architectures (SPAA)},
  2008.

\bibitem{KoutisMP_FOCS10}
Ioannis Koutis, Gary~L. Miller, and Richard Peng.
\newblock {Approaching optimality for solving SDD systems}.
\newblock In {\em FOCS '10: Proceedings of the 51st Annual IEEE Symposium on
  Foundations of Computer Science}. IEEE Computer Society, 2010.

\bibitem{KoutisMP_FOCS11}
Ioannis Koutis, Gary~L. Miller, and Richard Peng.
\newblock {A nearly $m\log n$ solver for SDD linear systems}.
\newblock In {\em FOCS '11: Proceedings of the 52nd Annual IEEE Symposium on
  Foundations of Computer Science}. IEEE Computer Society, 2011.

\bibitem{KMP2012}
Ioannis Koutis, Gary~L. Miller, and Richard Peng.
\newblock A fast solver for a class of linear systems.
\newblock {\em Commun. ACM}, 55(10):99--107, October 2012.

\bibitem{KrishnanFS13}
Dilip Krishnan, Raanan Fattal, and Richard Szeliski.
\newblock Efficient preconditioning of laplacian matrices for computer
  graphics.
\newblock {\em ACM Trans. Graph.}, 32(4):142, 2013.

\bibitem{LivneB12}
Oren~E. Livne and Achi Brandt.
\newblock {Lean Algebraic Multigrid ({LAMG}): Fast Graph Laplacian Linear
  Solver}.
\newblock {\em SIAM J. Scientific Computing}, 34(4), 2012.

\bibitem{DBLP:conf/soda/OrecchiaV11}
Lorenzo Orecchia and Nisheeth~K. Vishnoi.
\newblock Towards an {SDP}-based approach to spectral methods: A
  nearly-linear-time algorithm for graph partitioning and decomposition.
\newblock In Dana Randall, editor, {\em SODA}, pages 532--545. SIAM, 2011.

\bibitem{Peng.thesis}
Richard Peng.
\newblock {\em Algorithm design using spectral graph theory}.
\newblock PhD thesis, Carnegie Mellon University, 2013.

\bibitem{PengSpielman13}
Richard Peng and Daniel~A. Spielman.
\newblock An efficient parallel solver for {SDD} linear systems.
\newblock {\em CoRR}, abs/1311.3286, 2013.

\bibitem{SpielmanS08}
Daniel~A. Spielman and Nikhil Srivastava.
\newblock Graph sparsification by effective resistances.
\newblock In {\em Proceedings of the 40th Annual ACM Symposium on Theory of
  Computing (STOC)}, pages 563--568, 2008.

\bibitem{SpielmanTeng04}
Daniel~A. Spielman and Shang-Hua Teng.
\newblock Nearly-linear time algorithms for graph partitioning, graph
  sparsification, and solving linear systems.
\newblock In {\em Proceedings of the 36th Annual ACM Symposium on Theory of
  Computing (STOC)}, pages 81--90, June 2004.

\bibitem{Teng:2010}
Shang-Hua Teng.
\newblock The laplacian paradigm: emerging algorithms for massive graphs.
\newblock In {\em Proceedings of the 7th annual conference on Theory and
  Applications of Models of Computation}, TAMC'10, pages 2--14, Berlin,
  Heidelberg, 2010. Springer-Verlag.

\bibitem{ThorupZ05}
Mikkel Thorup and Uri Zwick.
\newblock Approximate distance oracles.
\newblock {\em J. ACM}, 52(1):1--24, 2005.

\bibitem{tropp2012user}
Joel~A Tropp.
\newblock User-friendly tail bounds for sums of random matrices.
\newblock {\em Foundations of Computational Mathematics}, 12(4):389--434, 2012.

\end{thebibliography}

\end{document}